\documentclass[submission]{eptcs}
\providecommand{\event}{AFL 2014}

\usepackage{amsmath}
\usepackage{amsthm}
\usepackage{amsfonts}
\usepackage{amssymb}

\usepackage{graphicx}
\usepackage{epic}
\usepackage{eepic}
\usepackage{epsfig,float}
\usepackage{multicol}
\usepackage{pdfsync}
\usepackage{color}
\usepackage{verbatim}
\usepackage{paralist}

\newcommand{\Sig}{\Sigma}
\newcommand{\eps}{\varepsilon}

\newcommand{\rev}{R}
\newcommand{\deter}{D}

\newcommand{\cA}{\mathcal{A}}

\newcommand{\cB}{\mathcal{B}}
\newcommand{\cD}{\mathcal{D}}

\newcommand{\cN}{\mathcal{N}}

\newcommand{\emp}{\emptyset}
\newcommand{\qedb}{\hfill$\blacksquare$} 
\newcommand{\ol}{\overline}

\newcommand{\im}{\operatorname{im}}
\newcommand{\coim}{\operatorname{coim}}

\newcommand{\noin}{\noindent}

\newcommand{\atoms}{{\bf A}}
\newcommand{\atp}{{\phi(\atoms)}}

\newcommand{\be}{\begin{enumerate}}
\newcommand{\ee}{\end{enumerate}}
\newcommand{\bi}{\begin{itemize}}
\newcommand{\ei}{\end{itemize}}

\newtheorem{theorem}{Theorem}
\newtheorem{lemma}{Lemma}
\newtheorem{proposition}{Proposition}

\theoremstyle{remark}

\newtheorem{example}{Example}
\newtheorem{remark}{Remark}

\title{Maximally Atomic Languages\thanks{This work was supported by the Natural Sciences and Engineering Research Council of Canada under grant No.~OGP0000871.}}
\author{Janusz Brzozowski
\institute{David R. Cheriton School of Computer Science, University of Waterloo\\
Waterloo, ON, Canada N2L 3G1}
\email{brzozo@uwaterloo.ca}\smallskip\\
Gareth Davies
\institute{Department of Pure Mathematics, University of Waterloo\\
Waterloo, ON, Canada N2L 3G1}
\email{gdavies@uwaterloo.ca}
}
\def\titlerunning{Maximally Atomic Languages}
\def\authorrunning{J. Brzozowski \& G. Davies}

\begin{document}
\maketitle

\begin{abstract}
The atoms of a regular language are non-empty intersections of complemented and uncomplemented quotients of the language.
Tight upper bounds on the number of atoms of a language and on the quotient complexities of atoms are known.
We introduce a new class of regular languages, called the \emph{maximally atomic languages}, consisting of all languages meeting these bounds.
We prove the following result:
If $L$ is a regular language of quotient complexity $n$ and $G$ is the subgroup of permutations in the transition semigroup $T$ of the minimal DFA of $L$, then 
$L$ is maximally atomic if and only if $G$ is transitive on $k$-subsets of $\{1,\dotsc,n\}$ for $0 \le k \le n$ and $T$ contains a transformation of rank $n-1$.
\medskip

\noin
{\bf Keywords:}
atom, \'atomaton, finite automaton, quotient complexity, regular language, set-transitive group, state complexity, transition semigroup

\end{abstract}

\section{Introduction}
The \emph{state/quotient complexity} of a regular language is the number of states in the minimal deterministic finite automaton (DFA) of the language, or equivalently, the number of left quotients of the language.
An \emph{atom} of a regular language is a non-empty intersection of the language's left quotients, some of which may be complemented.
Brzozowski and Tamm have found tight upper bounds on the number of atoms of a language~\cite{BrTa14} and on the quotient complexities of atoms~\cite{BrTa13}.
This lets us define a new class of regular languages  which we call \emph{maximally atomic}: these are  regular languages whose atoms meet these bounds.

The \emph{transition semigroup} of a DFA is the semigroup of transformations induced by the transition function of the DFA  on its set of states. Our main result (stated formally in Section \ref{sec:main}) is the following relationship between maximally atomic languages and transition semigroups:
\smallskip

\em
A regular language with quotient complexity $n$ is maximally atomic if and only if the transition semigroup of its minimal DFA contains permutations that can map any subset of $\{1,\dotsc,n\}$ to any other subset of the same size, as well as at least one transformation with an image of size $n-1$.
\smallskip

\rm
In the process of proving this, we establish several other relationships between transition semigroups and atoms; in particular, we give
sufficient conditions for a language to have the maximal number of atoms, and necessary and sufficient conditions for certain individual atoms to have maximal complexity. We also derive a general formula for the transition functions of ``\'atomata'' (nondeterministic finite automata whose states correspond to the atoms of the language they recognize).

\section{Definitions and Terminology}

\subsection{Partially Ordered Sets}
A \emph{partially ordered set} (poset) is a pair $(S,\le)$ where $S$ is a set and $\le$ is a partial order on $S$. 
A \emph{subposet} of $(S,\le)$ is a poset $(T,\le)$ such that $T \subseteq S$. 
We often abbreviate $(S,\le)$ to simply $S$. 

If $T$ is a subposet of $S$, then for $a,b \in S$, the \emph{interval} of $T$ between $a$ and $b$, denoted $[a,b]_T$, is the set of all $t \in T$ such that $a \le t$ and $t \le b$. Note that if $b < a$, then the interval $[a,b]_T$ is empty.

Let $Q_n = \{1,2,\dotsc,n\}$, let $P = (2^{Q_n},\subseteq)$, and let $X$ be a subposet of $P$.
For each non-empty interval $[V,U]_X$, define the \emph{type} of $[V,U]_X$ to be the pair of integers $\left(\left|\bigcap [V,U]_X \right|,\left|\bigcup [V,U]_X\right|\right)$.
Let the type of the empty interval be $(-1,-1)$.

\subsection{Transformations}
A \emph{transformation} of a set $X$ is a mapping $t \colon X \rightarrow X$.  Since we deal only with finite sets, we assume without loss of generality that  $X = Q_n$ for some $n$.
A \emph{permutation} is an invertible (one-to-one and onto) transformation. 
A \emph{singular transformation}
is a non-invertible transformation.

If a transformation $t$ maps $i$ to $j$, we say the \emph{image} of $i$ under $t$ is $j$ and write $t(i) = j$. 
The image of $S \subseteq Q_n$ is $t(S) = \{ t(i) \mid i \in S \}$. 
The image of $t$ itself is $\im t = t(Q_n)$.
The \emph{coimage} of $t$ is $\coim t = \ol{\im t}$, where $\ol{S} = Q_n \setminus S$. 
The \emph{preimage} of an element $i$ under $t$ is $t^{-1}(i) = \{ j \mid t(j) = i \}$. 
The preimage of $S \subseteq Q_n$ under $t$ is $t^{-1}(S) = \bigcup_{i \in S} t^{-1}(i)$. 
The \emph{rank} of a transformation is $|\im t|$. 
The \emph{composition} or \emph{product} of two transformations $s$ and $t$ is $s \circ t$, defined by $(s \circ t)(i) = s(t(i))$. 

A \emph{transposition} $(i,j)$ for $i\neq j$  is a transformation such that $t(i) = j$, $t(j) = i$, and $t(\ell) = \ell$ for all  $\ell \not\in \{i,j\}$. 
A permutation is  \emph{even} if it can be written as a product of an even number of transpositions and  it is  \emph{odd} otherwise.
A~\emph{unitary transformation}, denoted  by $(i \rightarrow j)$ (with $i \ne j$), is a transformation such that $t(i) = j$ and $t(\ell) = \ell$ for all $\ell \ne i$.

\subsection{Semigroups, Monoids, and  Groups }
A \emph{semigroup} is a pair $(S,\cdot)$, where $S$ is a non-empty set and $\cdot$ is an associative binary operation. 
We often abbreviate $(S,\cdot)$ to $S$.
A \emph{monoid} $M = (M,\cdot,e)$ is a semigroup with identity $e$, and a \emph{group} $G = (G,\cdot,e)$ is a monoid in which each element has an inverse.
A \emph{subsemigroup} of $(S,\cdot)$ is a semigroup $(T,\cdot)$ where $T \subseteq S$.
If $(S,\cdot,e)$ and $(M,\cdot,e)$ are monoids with $M \subseteq S$, then
$M$ is a \emph{submonoid} of $S$.
A \emph{subgroup} of $S$ is a submonoid $G$ of $S$ such that $G$ is a group.

The \emph{full transformation semigroup} of degree $n$, denoted $T_n$, is the set of all transformations $t \colon Q_n \rightarrow Q_n$ under the binary operation $\circ$. 
Note that $T_n$ is a monoid, since the identity transformation of $Q_n$ acts as the identity element.
The \emph{symmetric group} of degree $n$, denoted  by  $S_n$, is the subgroup of permutations in $T_n$.
A \emph{transformation semigroup} of degree $n$ is a subsemigroup of $T_n$, and a \emph{permutation group} of degree $n$ is a subgroup of $S_n$.
A \emph{conjugate} of a permutation group $G$ of degree $n$ is a group of the form $\{ p \circ g \circ p^{-1} \mid g \in G\}$, where $p \in S_n$.

Let $G$ be a permutation group of degree $n$ and  let $X$ be   a set.
For $x \in X$, the \emph{orbit} of $x$ under $G$ is the set $\{ g(x) \mid g \in G \}$. 
We say that $G$ \emph{acts transitively} or \emph{is transitive} on a set $X$ if for all $x,y \in X$ there exists $g \in G$ such that $g(x) = y$, or equivalently, if $G$ has only one orbit when it acts on $X$.
We say $G$ is \emph{$k$-set-transitive} if it is transitive on the set of $k$-subsets (subsets of cardinality $k$) of $Q_n$.
If $G$ is $k$-set-transitive for $0 \le k \le n$, we say $G$ is \emph{set-transitive}. 

The set-transitive permutation groups have been fully classified by Beaumont and Peterson~\cite{BePe55}.
In general, a set-transitive group is either the symmetric group $S_n$ or the alternating group $A_n$ (the subgroup of even permutations in $S_n$).
When $n$ is small there are four exceptions (up to conjugation):
\begin{proposition}
\label{prop:settrans}
A set-transitive permutation group of degree $n$ is $S_n$ or $A_n$ or a conjugate of one of the following permutation groups:
\be
\item
For $n=5$, the affine general linear group $\operatorname{AGL}(1,5)$.
\item
For $n=6$, the projective general linear group $\operatorname{PGL}(2,5)$.
\item
For $n=9$, the projective special linear group $\operatorname{PSL}(2,8)$.
\item
For $n=9$, the projective semilinear group $\operatorname{P\Gamma L}(2,8)$.
\ee
\end{proposition}

\subsection{Finite Automata}
A \emph{nondeterministic finite automaton} (\emph{NFA}) is a tuple $\cN = (Q,\Sig,\eta,I,F)$, where $Q$ is a  finite, non-empty  set of \emph{states}, $\Sig$ is a finite, non-empty   \emph{alphabet}, $\eta\colon Q \times \Sig \rightarrow 2^Q$ is a \emph{transition function}, $I \subseteq Q$ is a set of \emph{initial} states, and $F \subseteq Q$ is a set of \emph{final} states.
We extend $\eta$ to $\eta\colon 2^Q \times \Sig^* \rightarrow 2^Q$ as follows: for $S \subseteq Q$ and $w = xa$, $x \in \Sig^*$, $a \in \Sig$, we define $\eta(S,w)$ inductively by $\eta(S,\eps) = S$ and $\eta(S,xa) = \eta(\eta(S,x),a) = \bigcup_{s \in \eta(S,x)} \eta(s,a)$. 
We define $\eta_w \colon 2^Q \rightarrow 2^Q$ by $\eta_w(S) = \eta(S,w)$. 

A word $w$ is \emph{accepted} by $\cN$ if $\eta_w(I) \cap F \ne \emp$. 
The \emph{language accepted by $\cN$} is the set of all words accepted by $\cN$. 
The \emph{language of a state $q \in Q$} is the language accepted by the modified NFA $\cN_q = (Q,\Sig,\eta,\{q\},F)$. 
For $S,T \subseteq Q$, we say $S$ is \emph{reachable from $T$} in $\cN$ if there exists $w \in \Sig^*$ such that $\eta_w(T) = S$. 
If $S$ is reachable from $I$, we  simply say $S$ is   \emph{reachable} in $\cN$.
An NFA that accepts a language $L$ is \emph{minimal} if the number of states is minimal among all NFAs that accept $L$.

A \emph{deterministic finite automaton} (\emph{DFA}) is a tuple $\cD = (Q,\Sig,\delta,q_1,F)$, where $Q$, $\Sig$ and $F$ have the same meaning as in an NFA, $\delta \colon Q \times \Sig \rightarrow Q$ is a transition function, and $q_1 \in Q$ is an initial state.
Since DFAs are special cases of NFAs, all the definitions above apply  also  to DFAs. 
While minimal NFAs need not be unique, there is a unique (up to isomorphism) minimal DFA for each regular language. 

For all $w \in \Sig^*$, $\delta_w \colon Q \rightarrow Q$ is a transformation of the set of states of $\cD$; we call this the \emph{transformation induced by $w$} in $\cD$.
The \emph{transition semigroup} of $\cD$ is the semigroup $(T,\circ)$, where $T = \{\delta_w \mid w \in \Sig^+\}$.
This is the semigroup of transformations of $Q$ induced by non-empty words over $\Sig$ in $\cD$.

For an NFA $\cN = (Q,\Sig,\eta,I,F)$, define the \emph{reverse} of $\cN$ to be the NFA $\cN^\rev = (Q,\Sig,\eta^\rev,F,I)$, where $\eta^\rev(q,a) = \{ p \in Q \mid q \in \eta(p,a) \}$.
Note that if $\cN = \cD$ is a DFA with transition function $\delta$, then $\delta_w$ is a transformation and we have $\delta_w^\rev = \delta_w^{-1}$.
Define the \emph{determinization} of an NFA $\cN$ to be the DFA $\cN^\deter = (Q',\Sig,\eta^\deter,I,F')$, where $Q' = \{S \in 2^Q \mid \text{$S$ is reachable in $\cN$} \}$, $F' = \{ S \in Q' \mid S \cap F \ne \emp\}$, and $\eta^\deter(S,a) = \bigcup_{s \in S} \eta(s,a)$. 

\subsection{Languages, Quotients, and Atoms}
Let $L$ be a regular language over the alphabet $\Sig$ and let $\cD = (Q_n,\Sig,\delta,q_1,F)$ be the minimal DFA of $L$. 
The \emph{left quotient} (or simply \emph{quotient}) of $L$ by the word $w \in \Sig^*$ is $w^{-1}L = \{ x \mid wx \in L\}$.
There is a one-to-one correspondence between quotients of $L$ and states of the minimal DFA of $L$: the languages of distinct states of $\cD$ are distinct quotients of $L$.
We use the following convention when discussing quotients of $L$: the set of quotients is $\{K_1,K_2,\dotsc,K_n\}$, where $K_i$ is the language of state $i$ of $\cD$. 
Due to the one-to-one correspondence between states and quotients, the \emph{complexity} of $L$ can be equivalently defined as the number of states in the minimal DFA of $L$ (\emph{state complexity}) or the number of distinct quotients of $L$ (\emph{quotient complexity}).

From now on we deal with non-empty languages only.
Denote the complement of a language $L$ by $\ol{L} = \Sig^* \setminus L$.
For $S \subseteq Q_n$, let $A_S$ denote the intersection $\bigcap_{i \in S} K_i \cap \bigcap_{i \in \ol{S}} \ol{K_i}$.
If $A_S$ is non-empty, then $A_S$ is called an \emph{atom} of $L$. 
Let $\atoms$ be the set of all atoms of $L$. 
The \emph{atom map} $\phi\colon \atoms \rightarrow 2^{Q_n}$ is defined by $\phi(A_S) = S$. 
This map is well-defined, since for each atom $A$ there is precisely one subset $S$ of $Q_n$ such that $A_S = A$.
The \emph{basis} of an atom $A$ is $\cB(A) = \{ K_i \mid i \in \phi(A) \}$.

The \emph{\'atomaton} of $L$ is the NFA $\cA = (\atoms,\Sig,\eta,I,F)$, where 
$\eta(A_i,a) = \{ A_j \mid aA_j \subseteq A_i \}$,
$I = \{ A \in \atoms \mid q_1 \in \phi(A) \}$, 
and $F = \{ A \in \atoms \mid \eps \in A \}$. 
Note that the initial atoms are those that contain $L$ in their bases. 
Also, there is precisely one final atom: the atom  for which all the quotients in its basis contain $\eps$ and all other quotients do not.
The language of state $A$ of $\cA$ is the atom $A$~\cite{BrTa14}. 

The \emph{atomic poset} of $L$ is $\atp = (\phi(\atoms),\subseteq)$; this is the set of all subsets $S$ of $Q_n$ such that $A_S$ is an atom. An \emph{atomic interval} of $L$ is an interval in $\atp$, that is, an interval of the form $[V,U]_\atp$. We denote an atomic interval using double brackets, since this makes the notation cleaner: we write $[[V,U]]$ instead of $[V,U]_\atp$. Since $\atp$ is a subposet of $(2^{Q_n},\subseteq)$, any two subsets of $Q_n$ can act as endpoints of an atomic interval. Furthermore, every atomic interval $[[V,U]]$ has an associated type $(v,u)$, as defined in the section on posets.

Note that, if $[[V,U]]$ contains both of its endpoints (i.e., $V,U \in [[V,U]]$), then the type of $[[V,U]]$ is $(|V|,|U|)$.
However, we cannot always use the sizes of the endpoints to determine the type of an interval, since there may be multiple ways to choose the endpoints of an interval. 
For example, if $A_{\{1\}}$ is an atom but $A_\emp$ and $A_{\{1,2\}}$ are not, then $[[\{1\},\{1\}]] = [[\emp,\{1\}]] = [[\{1\},\{1,2\}]] = \{\{1\}\}$. 
But this interval has type $(1,1)$, not $(0,1)$ or $(1,2)$.

Some basic facts about atoms and \'atomata follow. The following proposition, proved in \cite{BrTa13}, shows that we may view the states of $\cA$ as subsets of $Q_n$:

\begin{proposition}
\label{prop:atomiso}
Let $L$ be a regular language with \'atomaton $\cA$ and minimal DFA $\cD$. Then the atom map $\phi$ is an NFA isomorphism between $\cA$ and $\cD^{\rev\deter\rev}$.
\end{proposition}

The next proposition relates the number of atoms of $L$ to the complexity of the reverse $L^\rev$. The proof follows easily from Proposition \ref{prop:atomiso}.

\begin{proposition}[Number of Atoms]
\label{prop:atomrev}
Let $L$ be a regular language with complexity $n$, and let  the minimal DFA of $L$ be $\cD=(Q_n,\Sig,\delta,q_1,F)$. 
Then for $S \subseteq Q_n$, the intersection $A_S$ is an atom of $L$ if and only if $S$ is reachable in $\cD^\rev$, i.e, if and only if there exists $w \in \Sig^*$ such that $\delta^{-1}_w(F) = S$. 
Thus there is a bijection between atoms of $L$ and states of $\cD^{\rev\deter}$, the minimal DFA of $L^\rev$.
\end{proposition}

It is well-known that if the complexity of $L$ is $n$, then the complexity of $L^{\rev}$ is at most $2^n$, and for $n \ge 2$ this bound is tight.
Thus $2^n$ is also a tight bound on the number of atoms of a regular language when $n \ge 2$.

In~\cite{BrTa13}, a tight upper bound on the complexity of individual atoms was derived and a formula for the bound was given.
We give a different (but equivalent) formula below:

\begin{proposition}[Complexity of Atoms]
\label{prop:atombounds}
Let $L$ be a regular language with complexity $n$.
Define the function $\Psi$ as follows:
\[
\Psi(n,k) 
= \begin{cases}
2^n - 1, & \text{if $k = 0$ or $k = n$;} \\
1 + \sum_{v = 1}^k \sum_{u = k}^{n-1} \binom{n}{u}\binom{u}{v}, & \text{if $1 \le k \le n-1$.} \\
\end{cases}
\]
If $A_S$ is an atom of $L$, $\Psi(n,|S|)$ is a tight upper bound on the complexity of $A_S$.
\end{proposition}

With these bounds established, we can formally define the class of maximally atomic languages.
A~non-empty regular language $L$ of complexity $n$ is \emph{maximally atomic} if it has the maximal number of atoms (1 if $n=1$, $2^n$ if $n \ge 2$) and if for each atom $A_S$ of $L$, $A_S$ has the maximal complexity $\Psi(n,|S|)$.

\section{Main Results}
\label{sec:main}

Note that when $n = 1$, the only nonempty language over $\Sig$ is $\Sig^*$, and it is maximally atomic.
The following proposition characterizes the maximally atomic languages of complexity $n = 2$:
\begin{proposition}
Let $L$ be a regular language of complexity $2$ and let $\cD$ be its minimal DFA with state set $Q_2$.
Let $T$ be the transition semigroup of $\cD$. 
Then:
\bi
\item
There are four transformations of $Q_2$: the identity transformation, the transposition $(1,2)$, and the unitary transformations $(1 \rightarrow 2)$ and $(2 \rightarrow 1)$.
\item
$T$ contains all four transformations of $Q_2$ if and only if $T$ contains $(1,2)$ and at least one unitary transformation.
\item
All subsets of $Q_2$ are reachable in $\cD^\rev$ (and hence $L$ has all $2^2$ atoms) if and only if $T$ contains all four transformations of $Q_2$.
\item
Each atom of $L$ has maximal complexity if and only if $T$ contains all four transformations of $Q_2$.
\item
Thus, $L$ is maximally atomic if and only if $T$ contains all four transformations of $Q_2$.
\ei
\end{proposition}
The computations required to prove this proposition can be easily done by hand.
Henceforth we will be concerned only  with languages of complexity $n \ge 3$.
\smallskip

Our main theorem is the following:
\begin{theorem}
\label{thm:maxatom}
Let $L$ be a regular language over $\Sig$ with complexity $n \ge 3$, and let $T$ be the transition semigroup of the minimal DFA of $L$. Then $L$ is maximally atomic if and only if the subgroup of permutations in $T$ is set-transitive and $T$ contains a transformation of rank $n-1$.
\end{theorem}

In view of this, let us consider how the class of maximally atomic languages relates to other language classes.
Let \textbf{FTS} denote the class of languages whose minimal DFAs have the \emph{full transformation semigroup} as their transition semigroup, 
let 
\textbf{STS} denote the class whose minimal DFAs have transition semigroups with a \emph{set-transitive subgroup} of permutations and a transformation of rank $n-1$, 
let \textbf{MAL} denote the class of \emph{maximally atomic languages}, let \textbf{MNA} denote the class of languages with the \emph{maximal number of atoms}, and let \textbf{MCR} denote the class of languages with a \emph{maximally complex reverse}. 
\be
\item
\textbf{FTS} is properly contained in \textbf{STS}, by Proposition \ref{prop:settrans}.
\item
\textbf{STS} is equal to \textbf{MAL}, by Theorem \ref{thm:maxatom}.
\item
\textbf{MAL} is contained in \textbf{MNA}. 
Figure 1 in \cite{BrDa13} shows the containment is proper.

\item
\textbf{MNA} is equal to \textbf{MCR}, by Proposition \ref{prop:atomrev}.
\ee
To summarize, we have:
$ \text{\textbf{FTS} $\subset$ \textbf{STS} = \textbf{MAL} $\subset$ \textbf{MNA} = \textbf{MCR}.} $

The proof of Theorem \ref{thm:maxatom} relies on two intermediate results.
The first gives a condition that is sufficient (but not necessary) for $L$ to have $2^n$ atoms:
\begin{theorem}
\label{thm:atomnum}
Let $L$ be a regular language over $\Sig$ with complexity $n \ge 3$, and let $T$ be the transition semigroup of the minimal DFA of $L$. If $T$ contains all unitary transformations, then $L$ has $2^n$ atoms.
\end{theorem}

The second result establishes Theorem \ref{thm:maxatom} in all but a few cases; it gives necessary and sufficient conditions for individual atoms of $L$ to have maximal complexity, but only when the bases of the atoms are in a certain size range.
\begin{theorem}
\label{thm:atomcomp}
Let $L$ be a regular language over $\Sig$ with complexity $n \ge 3$, and let $T$ be the transition semigroup of the minimal DFA of $L$.
Let $A_S$ be an atom of $L$ and suppose that either $n \ge 4$ and $2 \le |S| \le n-2$, or $n = 3$ and $1 \le |S| \le 2$. 
Then $A_S$ has maximal complexity if and only if the subgroup of permutations in $T$ is $|S|$-set-transitive and $T$ contains a transformation of rank $n-1$.
\end{theorem}
The rest of the paper consists of the proofs of these three theorems.
Shortly before the deadline for this paper, we were informed that the proof of our main result can be simplified by replacing the \'atomaton with a different construction~\cite{Ivan14}. Below we present our original proofs, which use the \'atomaton.

\section{Proof of Theorem \ref{thm:atomnum}}
Let $L$ be a language of complexity $n \ge 3$ and  let $\cD = (Q_n,\Sig,\delta,q_1,F)$ be  its minimal DFA.
Let $T$ be the transition semigroup of $\cD$ and assume it contains all unitary transformations.
 By Proposition \ref{prop:atomrev}, $L$ has $2^n$ atoms if and only if for all $S \subseteq Q_n$, $S$ is reachable in $\cD^\rev$, i.e., there exists $w \in \Sig^*$ such that $\delta^\rev_w(F) = \delta^{-1}_w(F) = S$.

Suppose $X \subseteq Q_n$, with $1 \le |X| \le n-1$. Let $t = (i \rightarrow j)$ and $s = (i \rightarrow k)$ for $i \in Q_n$, $j \in X$, $k \not\in X$; then $t^{-1}(X) = X \cup \{i\}$ and $s^{-1}(X) = X \setminus \{i\}$. Since $T$ contains all unitary transformations, it contains $t$ and $s$. Thus for every non-empty $X \subset Q_n$ and every $i \in Q_n$, there are words $w,x \in \Sig^*$ such that $\delta_w^{-1}(X) = X \cup \{i\}$ and $\delta_x^{-1}(X) = X \setminus \{i\}$. 

In other words, from any non-empty proper subset $X$ of $Q_n$, we can reach (in $\cD^\rev$) all subsets that differ from $X$ by the addition or removal of a single element. Repeatedly applying this fact, we see that from $X$ we can reach any subset $S$ of $Q_n$: shrink $X$ to a singleton $\{i\} \subseteq X$, expand $\{i\}$ to $\{i,j\}$ for $j \in S$, shrink again to $\{j\} \subseteq S$, and then expand to $S$ (or shrink to $\emp$ for $S = \emp$).

Now, if $|F| = 0$ then $L = \emp$, and if $|F| = n$ then $L = \Sig^*$; since $\cD$ is minimal, $n = 1$ in either case.
Since $n \ge 3$, we have that $F$ is a non-empty proper subset of $Q_n$.
Thus by the argument above, we can reach all subsets of $Q_n$ in $\cD^\rev$; hence $L$ has $2^n$ atoms. \qed

\section{Proof of Theorem \ref{thm:atomcomp}}
\subsection{The \'Atomaton and Minimal DFAs of Atoms}
\label{sec:atomaton}
In this section we prove the $\Rightarrow$ direction of Theorem \ref{thm:atomcomp}. 
Two results on \'atomata and atoms are needed for this.
We first describe the transition function of the \'atomaton, in the case where the states are viewed as subsets of $Q_n$.
Define $\Delta_w \colon 2^{Q_n} \rightarrow 2^{Q_n}$ by $\Delta_w(S) = \ol{\delta_w(\ol{S})} = Q_n \setminus \delta_w(Q_n \setminus S)$.

\begin{lemma}
\label{lem:atrans}
Let $L$ be a regular language over $\Sig$. Let $\cD = (Q_n,\Sig,\delta,q_1,F)$ be the minimal DFA of $L$. Let $\cA$ be the \'atomaton of $L$ with transition function $\eta$. If $[[V,U]]$ is an atomic interval of $L$ and a set of states of $\cA$, then for all $w \in \Sig^*$, we have
$ \eta_w([[V,U]]) = [[\delta_w(V),\Delta_w(U)]]. $
\end{lemma}

\begin{proof}
It was shown in~\cite{BrTa13} that $\eta_a(S)=\{T\mid \text{$A_T$ is an atom of $L$, }T\supseteq \delta_a(S) \text{ and } \delta_a(\ol{S})\cap T=\emptyset\}$.
If $\delta_a(\ol{S}) \cap T = \emp$, then $T \subseteq \ol{\delta_a(\ol{S})} = \Delta_a(S)$. 
Thus $\eta_a(S)$ is the set of $T \subseteq Q_n$ such that $A_T$ is an atom of $L$ and $\delta_a(S) \subseteq T \subseteq \Delta_a(S)$, which is precisely $[[\delta_a(S),\Delta_a(S)]]$. One verifies that this can be extended to words, giving $\eta_w(S) = [[\delta_w(S),\Delta_w(S)]]$.

Next, we want to show $\eta_w([[V,U]]) = [[\delta_w(V),\Delta_w(U)]]$. 
For $T \in [[V,U]]$, consider $\eta_w(T)$. Since $V \subseteq T$,  $\delta_w(V) \subseteq \delta_w(T)$. Since $T \subseteq U$, we have $\ol{T} \supseteq \ol{U}$, and thus $\delta_w(\ol{T}) \supseteq \delta_w(\ol{U})$. It follows that $\Delta_w(T) \subseteq \Delta_w(U)$. Hence $\eta_w(T) = [[\delta_w(T),\Delta_w(T)]] \subseteq [[\delta_w(V),\Delta_w(U)]]$, and $\eta_w([[V,U]]) \subseteq [[\delta_w(V),\Delta_w(U)]]$. 

For containment in the other direction, suppose that $T$ is in $[[\delta_w(V),\Delta_w(U)]]$; then $\delta_w(V) \subseteq T \subseteq \Delta_w(U)$ and $A_T$ is an atom. 
Let $S = \delta_w^{-1}(T)$; then we claim $S \in [[V,U]]$.
Since $T \supseteq \delta_w(V)$, we have $\delta_w^{-1}(T) = S \supseteq V$.
If $i \in T \subseteq \Delta_w(U)$, then $i \not\in \delta_w(\ol{U})$.
Hence $\delta_w^{-1}(i)$ is disjoint from $\ol{U}$ for all $i \in T$, and so $\delta_w^{-1}(T)$ is disjoint from $\ol{U}$.
It follows that $\delta_w^{-1}(T) = S\subseteq U$.
It remains to show $A_S$ is an atom; but since $A_T$ is an atom, by Proposition \ref{prop:atomrev}, there exists $x \in \Sig^*$ such that $\delta_x^{-1}(F) = T$. Thus $S = \delta_w^{-1}(T) = \delta_w^{-1}(\delta_x^{-1}(F)) = \delta_{xw}^{-1}(F)$, so by Proposition \ref{prop:atomrev}, $A_S$ is also an atom.

Hence $S \in [[V,U]]$, and it follows that $\eta_w(S) = [[T,\Delta_w(S)]] \subseteq \eta_w([[V,U]])$. To complete the proof, we must show $\eta_w(S)$ is non-empty (and thus contains $T$) by showing that $T \subseteq \Delta_w(S) = \ol{\delta_w(\ol{\delta_w^{-1}(T))}}$. Observe that if $i \in T$, then $\delta_w^{-1}(i) \subseteq \delta_w^{-1}(T)$. Thus $\delta_w^{-1}(i) \cap \ol{\delta_w^{-1}(T)} = \emp$, and so $i \not\in \delta_w(\ol{\delta_w^{-1}(T)})$, which gives $i \in \Delta_w(S)$ as required.
Thus $T \in \eta_w(S) = [[T,\Delta_w(S)]]$,
and it follows that if $T \in [[\delta_w(V),\Delta_w(U)]]$, then $T \in \eta_w([[V,U]])$. 
This proves that the two intervals must be equal. 
\end{proof}

\renewcommand{\arraystretch}{1.2}
\begin{table}[bth]
\begin{minipage}[b]{0.2\linewidth}
\caption{$\cD$.}
\label{tab:D}
\begin{center}
$
\begin{array}{|c|c||c|}    
\hline
            &\ \delta\  &\ a \  \\
\hline  
\rightarrow &  1        &  2    \\
\hline  
\leftarrow  &  2        &  3    \\
\hline  
\leftarrow  &  3        &  4    \\
\hline  
            &  4        &  4    \\
\hline  
\end{array}
$
\end{center}
\end{minipage}
\hspace{0.2cm}
\begin{minipage}[b]{0.2\linewidth}
\caption{$\cD^R$.}
\label{tab:Dr}
\begin{center}
$
\begin{array}{|c|c||c|}    
\hline
            &\ \delta^\rev\ &\ a \      \\
\hline  
\leftarrow  &  1            &           \\
\hline  
\rightarrow &  2            &  \{1\}    \\
\hline  
\rightarrow &  3            &  \{2\}    \\
\hline  
            &  4            &  \{3,4\}  \\
\hline  
\end{array}
$
\end{center}
\end{minipage}
\hspace{0.2cm}
\begin{minipage}[b]{0.25\linewidth}
\caption{$\cD^{RD}$.}
\label{tab:Drd}
\begin{center}
$
\begin{array}{|c|c||c|}    
\hline
            & \delta^{\rev\deter}   & a         \\
\hline  
\rightarrow & \{2,3\}               & \{1,2\}   \\
\hline  
\leftarrow  & \{1,2\}               & \{1\}     \\
\hline  
\leftarrow  & \{1\}                 & \emp      \\
\hline  
            & \emp                  & \emp      \\
\hline  
\end{array}
$
\end{center}
\end{minipage}
\hspace{0.2cm}
\begin{minipage}[b]{0.25\linewidth}
\caption{$\cA$.}
\label{tab:A}
\begin{center}
$
\begin{array}{|c|c||c|}    
\hline
            & \eta      & a                 \\
\hline  
\leftarrow  & \{2,3\}   &                   \\
\hline  
\rightarrow & \{1,2\}   & \{\{2,3\}\}       \\
\hline  
\rightarrow & \{1\}     & \{\{1,2\}\}       \\
\hline  
            & \emp      & \{\emp,\{1\}\}    \\
\hline  
\end{array}
$
\end{center}
\end{minipage}

\end{table}
 
\begin{example}
\label{ex:eta}
The minimal DFA $\cD$ of Table~\ref{tab:D} accepts the language $\{a,aa\}$. 
The NFA $\cD^R$ is in Table~\ref{tab:Dr} and the DFA $\cD^{RD}$,  in Table~\ref{tab:Drd}. The \'atomaton $\cA$ is  in Table~\ref{tab:A}.
In NFAs $\cD^R$ and $\cA$, a blank in an entry $(q,a)$ indicates that there is no transition from $q$ under $a$. However, when determinization is used in Table~\ref{tab:Drd},  the empty set of states of $\cD^R$ becomes a state of the resulting DFA $\cD^{RD}$.
A right arrow ($\rightarrow$) indicates an initial state and a left arrow ($\leftarrow$) indicates a final state.

Consider the atomic interval $[[\emp, \{1,2\}]]=\{\emp, \{1\}, \{1,2\} \}$; 
we have $\delta_a(\emp)=\emp$, and $\Delta_a(\{1,2\})=\ol{\delta_a(\ol{\{1,2\}})}
=\ol{\delta_a( \{3,4\})}=\ol{ \{4\}}=\{1,2,3\}$.
Thus to determine the result of $\eta_a([[\emp,\{1,2\}]])$, we take the interval
$[\emp,\{1,2,3\}]_{2^{Q_4}} =
\{ \emp, \{1\}, \{2\}, \{3\}, \{1,2\},\{1,3\}, \{2,3\}, \{1,2,3\} \}$ 
and we remove the sets that do not represent atoms. After this removal, we get
$\{ \emp, \{1\}, \{1,2\}, \{2,3\} \} $.
Hence $\eta_a( \{\emp, \{1\}, \{1,2\} \} )=\eta_a( [[\emp,\{1,2\}]])= [[\emp,\{1,2,3\}]]=\{ \emp, \{1\}, \{1,2\}, \{2,3\} \}$.
\qedb
\end{example}
\medskip

\begin{remark}
If we treat the set of states of $\cA$ as a subset of $2^{Q_n}$, then it is possible that the empty set is a state of $\cA$, as in Example~\ref{ex:eta}. Since we use the same symbol for $\eta$ and its extension to subsets of states, an ambiguity arises when $\eta$ is applied to the empty set. Specifically, $\eta_w(\emp)$ may mean ``$\eta_w$ applied to the state $\emp \in 2^{Q_n}$'', in which case $\eta_w(\emp) = \eta_w([[\emp,\emp]]) = [[\emp,\coim\delta_w]]$, or it may mean ``$\eta_w$ applied to the empty subset of states $\emp \subseteq 2^{Q_n}$'', in which case $\eta_w(\emp) = \emp$.
We avoid this ambiguity  by adopting the  convention that  $\eta_w(\emp)$ always means ``$\eta_w$ applied to $\emp \subseteq 2^{Q_n}$'' and  $\eta_w([[\emp,\emp]])$ has the other meaning.
\end{remark}

A corollary of this is that every reachable subset of states in the \'atomaton $\cA = (\atoms,\Sig,\eta,I,F)$ is an atomic interval of $L$. The same holds for every reachable subset of states in the NFA $\cA_S = (\atoms,\Sig,\eta,\{S\},F)$ recognizing the atom $A_S$. 
Since the determinization $\cA^\deter_S$ is the minimal DFA of $A_S$~\cite{BrTa13}, it follows that \emph{the states of minimal DFAs of atoms of $L$ may be represented as atomic intervals of $L$}.

If $A_S$ is an atom of $L$ with maximal complexity, certain restrictions apply to the \emph{types} of the atomic intervals in $\cA^\deter_S$. For $S \subseteq Q_n$, define an \emph{$S$-type} to be a pair of integers $(v,u)$ satisfying:
\be
\item
If $|S| = 0$, then $v = 0$ and $0 \le u \le n-1$.
\item
If $|S| = n$, then $1 \le v \le n$ and $u = n$.
\item
If $1 \le |S| \le n-1$, then $1 \le v \le |S|$ and $|S| \le u \le n-1$.
\ee
A non-empty interval that has an $S$-type is called an \emph{$S$-interval}.  The empty interval is a special case: it is an $S$-interval if and only if $1 \le |S| \le n-1$.   The significance of $S$-types and $S$-intervals is as follows:

\begin{lemma}
\label{lem:sintervals}
Let $L$ be a regular language with complexity $n$ and let $A_S$ be an atom of $L$. 
If $A_S$ has maximal complexity $\Psi(n,|S|)$, then the set of states of $\cA^\deter_S$ equals the set of atomic $S$-intervals of $L$.
\end{lemma}

\begin{proof}
A simple counting argument shows that the number of intervals of type $(v,u)$ in a subposet of $P = (2^{Q_n},\subseteq)$ is bounded from above by $\binom{n}{u}\binom{u}{v}$. Combining this fact with the definition of an $S$-type gives that $\Psi(n,|S|)$ is an upper bound on the number of $S$-intervals in a subposet of $P$.
Now, we know that $\cA^\deter_S$ has exactly $\Psi(n,|S|)$ states; if we show that these states are all atomic $S$-intervals of $L$, then this implies the state set of $\cA^\deter_S$ contains exactly $\Psi(n,|S|)$ distinct atomic $S$-intervals of $L$, and nothing else.
Since the atomic poset of $L$ is a subposet of $P$, there can be no more than $\Psi(n,|S|)$ atomic $S$-intervals of $L$, and this proves the result.
Thus we just need to show that every state of $\cA^\deter_S$ is an $S$-interval.

Let $[[V,U]]$ be a state of $\cA^\deter_S$ and suppose $[[\delta_w(S),\Delta_w(S)]] = [[V,U]]$. 
If $[[V,U]]$ is the empty interval, then it is automatically an $S$-interval for $1 \le |S| \le n-1$, by definition.
For $|S| = 0$ or $|S| = n$, the empty interval is not an $S$-interval, but this does not matter since it is not reachable in $\cA^\deter_S$. In the $|S| = 0$ case, a state of $\cA^\deter_S$ has the form $[[\delta_w(\emp),\Delta_w(\emp)]] = [[\emp,\coim\delta_w]]$, which always contains $\emp$ (since $A_S = A_\emp$ is an atom); thus every state of $\cA^\deter_S$ is a non-empty interval. For $|S| = n$, a similar argument works.

Next, suppose $[[V,U]]$ is non-empty. 
For this case, some setup is needed.
Define the set $X_S = \{[\delta_w(S),\Delta_w(S)]_P\mid w \in \Sig^*\}$ of intervals in $P$. One can verify that $(|\delta_w(S)|,|\Delta_w(S)|)$ is an $S$-type for all $S$ and $w$, and thus $X_S$ is a set of $S$-intervals of $P$.
This means $|X_S|$ is bounded from above by $\Psi(n,|S|)$.
Now, let $Y_S = \{[[\delta_w(S),\Delta_w(S)]] \mid w \in \Sig^*\}$; this is just the set of states of $\cA^\deter_S$, and thus it has size $\Psi(n,|S|)$. 
Define $\alpha\colon X_S \rightarrow Y_S$ by $\alpha([X,Y]_P) = [[X,Y]]$; this is clearly a surjection, and thus $|X_S| \ge |Y_S| = \Psi(n,|S|)$.
Since we also have $|X_S| \le \Psi(n,|S|)$, we get $|X_S| = |Y_S| = \Psi(n,|S|)$ and hence $\alpha$ is a bijection.

Now, assume without loss of generality that the type of $[[V,U]]$ is $(|V|,|U|)$.
Suppose for a contradiction that $|V| > |S|$.
Then since $[[V,U]] = [[\delta_w(S),\Delta_w(S)]]$, we have $\delta_w(S) \subset V$. We can find a set $X$ such that $|X| = |S|$ and $\delta_w(S) \subseteq X \subset V$.
Now, since $\delta_w(S) \subseteq X \subseteq \Delta_w(S)$, we have $X \in [[\delta_w(S),\Delta_w(S)]]$ if and only if $A_X$ is an atom.
But $X \not\in [[V,U]]$ since $X \subset V$, and thus $A_X$ is not an atom.
It follows that the interval $[[X,X]]$ is empty.
If $|S| = 0$ or $|S| = n$, then in fact $|S| = |X|$ and $[[S,S]]$ is clearly non-empty, a contradiction.
If $1 \le |S| \le n-1$, observe that $\alpha([X,X]_P) = [[X,X]] = \emp$.
But $\emp \in X_S$ since $\emp$ is an $S$-interval of $P$ for $1 \le |S| \le n-1$, so also $\alpha(\emp) = \emp$.
This is a contradiction, since $\alpha$ is a bijection.
Thus for $S$ of any size, we always have $|V| \le |S|$. 
A similar argument to the above shows that $|U| \ge |S|$.

Thus, if $|S| = 0$ we have $|V| = 0$ and $0 \le |U| \le |\Delta_w(S)| = n-1$.
If $|S| = n$ we have $1 \le |\delta_w(S)| \le |V| \le n$ and $|U| = n$.
If $1 \le |S| \le n-1$, then $1 \le |V| \le |S|$ and $|S| \le |U| \le n-1$.
Thus we have proved $(|V|,|U|)$ is an $S$-type. 
Hence every state $[[V,U]]$ of $\cA^\deter_S$ is an atomic $S$-interval of $L$, and the number of states equals the upper bound $\Psi(n,|S|)$ on the number of atomic $S$-intervals of $L$, proving the lemma. 
\end{proof}

Lemma \ref{lem:sintervals} has two particularly useful consequences. Let $A_S$ be an atom of maximal complexity, and suppose $V,U\subseteq Q_n$ are sets such that $(|V|,|U|)$ is an $S$-type. Then:
\be
\item
$[[V,U]]$ has type $(|V|,|U|)$. In particular, this means $[[V,U]]$ contains its endpoints $V$ and $U$.
\item
$[[V,U]]$ is a state of $\cA^\deter_S$.
\ee
(1) follows since $(|V|,|U|)$ is an $S$-type, and so if $[[V,U]]$ does not have type $(|V|,|U|)$, the number of atomic $S$-intervals of type $(|V|,|U|)$ is not maximal and hence $A_S$ is not maximally complex. (2) follows since if $[[V,U]]$ has the $S$-type $(|V|,|U|)$, it is an atomic $S$-interval and thus a state of $\cA^\deter_S$.

These facts are sufficient to prove one direction of Theorem \ref{thm:atomcomp}:

\begin{proof}[Theorem \ref{thm:atomcomp} ($\Rightarrow$ Direction)]
Let $L$ be a language of complexity $n \ge 3$, let $T$ be the transition semigroup of the minimal DFA of $L$, and let $A_S$ be an atom of $L$.
Suppose either $n = 3 $ and $1 \le |S| \le 2$, or $n \ge 4$ and $2 \le |S| \le n-2$.
We prove that if $A_S$ has maximal complexity, then the subgroup of permutations in $T$ is $|S|$-set-transitive and $T$ contains a transformation of rank $n-1$.

The minimal DFA of $A_S$ is $\cA^\deter_S$, and its initial state is $[[S,S]]$. 
For all $X \subseteq Q_n$ with $|X| = |S|$, $(|X|,|X|)$ is an $S$-type. Thus by Lemma \ref{lem:sintervals}, $[[X,X]]$ is a state of $\cA^\deter_S$ of type $(|X|,|X|)$.
Thus $\eta_w([[S,S]]) = [[\delta_w(S),\Delta_w(S)]] = [[X,X]]$ for some $w \in \Sig^*$.
Applying Lemma \ref{lem:sintervals} again gives $(|\delta_w(S)|,|\Delta_w(S)|) = (|X|,|X|)$.
Hence $|X| = |\delta_w(S)| = |S|$, and so $\delta_w \in T$ is a permutation. 
It follows for all $X \subseteq Q_n$ with $|X| = |S|$, there is a permutation that sends $S$ to $X$; thus the subgroup of permutations in $T$ is $|S|$-set-transitive.

Now, let $\delta_w \in T$ have rank $n-k$ and consider $[[\delta_w(S),\Delta_w(S)]]$.
By Lemma \ref{lem:sintervals} this interval has type $(|\delta_w(S)|,|\Delta_w(S)|)$, so it is a non-empty interval.
This implies $\delta_w(S)$ and $\delta_w(\ol{S})$ are disjoint.
It follows that
$|\im \delta_w| = |\delta_w(Q_n)| = |\delta_w(S)| + |\delta_w(\ol{S})|$.
Since the rank of $\delta_w$ is $n-k$, $|\delta_w(\ol{S})| = (n - k) - |\delta_w(S)|$. 
Thus $|\Delta_w(S)| = n - (n-k-|\delta_w(S)|) = |\delta_w(S)| + k$, which gives $|\Delta_w(S)| - |\delta_w(S)| = k$.

Consider $[[S,S\cup\{i\}]]$ for $i \not\in S$. Since $(|S|,|S|+1)$ is an $S$-type, by Lemma \ref{lem:sintervals} this interval is reachable in $\cA^\deter_S$.
Thus there is a $\delta_w \in T$ such that $(|\delta_w(S)|,|\Delta_w(S)|) = (|S|,|S|+1)$. By the argument above, this $\delta_w$ must have rank $n - (|\Delta_w(S)| - |\delta_w(S)|) = n-1$. Hence $T$ contains a transformation of rank $n-1$. 
\end{proof}

\subsection{Semigroups and Groups}
\label{sec:groups}
To prove the other direction of Theorem \ref{thm:atomcomp}, we use some results from semigroup and group theory. 
The first is a result of Livingstone and Wagner~\cite{LiWa65}:
\begin{proposition}
\label{prop:lw}
Let $G$ be a permutation group of degree $n \ge 4$. If $2 \le k \le \frac{n}{2}$, then the number of orbits when $G$ acts on $k$-subsets of $Q_n$ is at least the number of orbits when $G$ acts on $(k-1)$-subsets of $Q_n$.
\end{proposition}

 Using this proposition, we can easily prove 

\begin{lemma}
\label{lem:lw}
Let $G$ be a $k$-set-transitive permutation group of degree $n \ge 4$ and suppose $2 \le k \le \frac{n}{2}$. Then:
\be
\item
\label{lem:lw:1}
$G$ is $(n-k)$-set-transitive.
\item
\label{lem:lw:2}
$G$ is $\ell$-set-transitive for each $\ell$ such that $0 \le \ell \le k$ or $n-k \le \ell \le n$.
\ee
\end{lemma}
\begin{proof}
(\ref{lem:lw:1}): Suppose $G$ is $k$-set-transitive. If $U$ and $V$ are $(n-k)$-subsets of $Q_n$, then $\ol{U}$ and $\ol{V}$ are $k$-subsets, and  there exists a permutation $p \in G$ mapping $\ol{U}$ to $\ol{V}$. But if $p$ maps $\ol{U}$ to $\ol{V}$, then it maps $U$ to $V$; thus $G$ can map any $(n-k)$-subset to any other $(n-k)$-subset, and so is $(n-k)$-set-transitive.

(\ref{lem:lw:2}): Suppose $G$ is $k$-set-transitive and $2 \le k \le \frac{n}{2}$. Then there is one orbit when $G$ acts on $k$-subsets. By  Proposition \ref{prop:lw}, there is one orbit when $G$ acts on $(k-1)$-subsets. This implies $G$ is $(k-1)$-set-transitive. Repeating this argument we conclude that $G$ is $\ell$-set-transitive for $0 \le \ell \le k$. By (\ref{lem:lw:1}), $G$ is also $\ell$-set-transitive for $n-k \le \ell \le n$. 
\end{proof}

\noindent Note that for $n=3$,
a permutation group of degree 3 is set-transitive if and only if it is transitive.

The second result  we use is a theorem of Ru\v{s}kuc, published by McAlister~\cite{Mc98}:
\begin{proposition}
\label{prop:mca}
Let $G$ be a permutation group of degree $n \ge 3$ and let $t\colon Q_n \rightarrow Q_n$ be a unitary transformation. Let $T$ be the transformation semigroup generated by $G \cup \{t\}$. Then $T$ contains all singular transformations if and only if $G$ is 2-set-transitive.
\end{proposition}

We can use this to prove the following lemma:

\begin{lemma}
\label{lem:gens}
Let $G$ be a 2-set-transitive permutation group of degree $n \ge 3$ and let $t \colon Q_n \rightarrow Q_n$ be a transformation of rank $n-1$. Then the transformation semigroup $T$ generated by $G \cup \{t\}$ contains all singular transformations.
\end{lemma}

\begin{proof}
By Proposition \ref{prop:mca}, if $G$ is 2-set-transitive and $t$ is a unitary transformation, then $T$ contains all singular transformations. Thus it suffices to show that if $t$ is any transformation of rank $n-1$, then $G \cup \{t\}$ generates a unitary transformation.

For each transformation $s \colon Q_n \rightarrow Q_n$, we define a set of tuples called \emph{$s$-paths}.
 For $k \ge 2$,  a tuple $(i_1,\dotsc,i_k)$ of distinct elements of $Q_n$ is an \emph{$s$-path of length $k$} if $s(i_j) = i_{j+1}$ for $1 \le j < k$ and $s(i_k) = i_\ell$ for some $\ell < k$. 
An $s$-path $(i_1,\dotsc,i_k)$ is \emph{incomplete} if there exists $a$ in $Q_n$ such that $(a,i_1,\dotsc,i_k)$  is an $s$-path, and \emph{complete} otherwise. An $s$-path $(i_1,\dotsc,i_k)$ is \emph{cyclic} if $s(i_k) = i_1$ and \emph{acyclic} otherwise. The element $i_1$ of the acyclic $s$-path $(i_1,\dotsc,i_k)$ is called the \emph{head}.

Let $t$ be a transformation of rank $n-1$, and consider the $t$-paths. If a $t$-path is complete and acyclic, its head must be an element of $\coim t$. Since $t$ has rank $n-1$, $|\coim t| = 1$, and so there is precisely one complete acyclic $t$-path. Let $(a_1,\dotsc,a_k)$ be that complete acyclic $t$-path, and suppose $t(a_k) = a_\ell$. 

Since $G$ is 2-set-transitive, it is  1-set-transitive  by Lemma \ref{lem:lw}. 
Thus there exists a permutation $p \in G$ with $p(a_1) = a_{\ell-1}$. 
Let $pt = p \circ t$, and consider $pt$-paths. 
Since $pt$ has rank $n-1$, there is only one complete acyclic $pt$-path; the head of this path must be $a_{\ell-1}$, since $\coim pt = \{a_{\ell-1}\}$.
Observe that
$pt(a_k) = p(t(a_k)) = p(a_{\ell}) = p(t(a_{\ell-1})) = pt(a_{\ell-1})$;
it follows that $(a_{\ell-1},p(a_\ell),pt(p(a_\ell)),\dotsc,a_k)$ is the complete acyclic $pt$-path.

Now, let $n$ be the product of the lengths of all the complete cyclic $pt$-paths and the incomplete cyclic $pt$-path $(p(a_\ell),pt(p(a,\ell)),\dotsc,a_k)$. Then we have $(pt)^n = (a_{\ell-1} \rightarrow (pt)^{n}(a_{\ell-1}))$, where $(pt)^n$ is $pt$ composed with itself $n$ times.
This proves that $T$ must contain all singular transformations, since it is 2-set-transitive and contains a unitary transformation. 
\end{proof}

These results are sufficient to prove the other direction of Theorem \ref{thm:atomcomp}:
\begin{proof}[Theorem \ref{thm:atomcomp} ($\Leftarrow$ Direction)]
Let $L$ be a language of complexity $n \ge 3$, let $T$ be the transition semigroup of the minimal DFA of $L$, and let $A_S$ be an atom of $L$.
Suppose either $n = 3 $ and $1 \le |S| \le 2$, or $n \ge 4$ and $2 \le |S| \le n-2$.
We prove that if the subgroup of permutations in $T$ is $|S|$-set-transitive and $T$ contains a transformation of rank $n-1$, then $A_S$ has maximal complexity.

By Lemmas \ref{lem:lw} and \ref{lem:gens}, $T$ contains all singular transformations. By Theorem \ref{thm:atomnum}, $L$ has $2^n$ atoms.
From the proof  of Lemma \ref{lem:sintervals}, $\Psi(n,|S|)$ is a tight bound on the number of $S$-intervals in the atomic poset of $L$.
Since $L$ has $2^n$ atoms (the maximal possible), the number of atomic $S$-intervals of $L$ meets the bound $\Psi(n,|S|)$.
It remains to show that all these intervals are reachable in the minimal DFA $\cA^\deter_S$ of $A_S$.
From the inital state $[[S,S]]$ of $\cA^\deter_S$, we can reach the empty interval by $(i \rightarrow j)$ where $i \in S$ and $j \not\in S$; thus it suffices to consider non-empty intervals.

Let $[[V,U]]$ be a non-empty  atomic $S$-interval of $L$ with type $(|V|,|U|)$.
By the definition of an atomic $S$-interval, $1 \le |V| \le |S|$ and $|S| \le |U| \le n-1$ and $V \subseteq U$.
Thus there exists a set $X$ such that $|X| = |S|$ and $V \subseteq X \subseteq U$. 
Since the subgroup of permutations in $T$ is $|S|$-set-transitive, there is a permutation $\delta_w \in T$ that sends $S$ to $X$; thus $\eta_w([[S,S]]) = [[X,X]]$.
If $V = X = U$, we are done, so assume that $V \subset X$ or $X \subset U$.
If $V \subset X$ and $|V| \ge 2$, we can shrink the lower bound of $[[X,X]]$ as follows: select distinct $i,j \in Q_n$ such that $i \in X \setminus V$ and $j \in V$.
Since $T$ contains all unitary transformations, there is a $\delta_x \in T$ such that $\delta_x = (i \rightarrow j)$. 
Since $i \not\in \ol{X}$, $\delta_x(\ol{X}) = \ol{X}$ and thus $\Delta_x(X) = X$.
It follows that $\eta_x([[X,X]]) = [[X \setminus \{i\},X]]$.
Repeating this process, we can reach $[[V,X]]$ for all $V$ with $1 \le |V| \le |S|$.
By a similar process, we can repeatedly enlarge the upper bound of $[[V,X]]$ to reach $[[V,U]]$.
Thus all $\Psi(n,|S|)$ atomic $S$-intervals of $L$ are reachable in $\cA^\deter_S$.
By Lemma \ref{lem:sintervals}, $A_S$ has maximal complexity. 
\end{proof}

\begin{remark}
\label{rem:unitary}
The proof above works for the $|S| = 1$ and $|S| = n-1$ cases if we 
assume that $T$ contains all unitary transformations, rather than only assuming it contains some transformation of rank $n-1$. 
\end{remark}

\section{Proof of Theorem \ref{thm:maxatom}}
\label{sec:proofs}
Having proved Theorems \ref{thm:atomnum} and \ref{thm:atomcomp}, we need  only a bit more work to prove our main theorem.

Let $L$ be a language with complexity $n \ge 3$ and let $T$ be the transition semigroup of the minimal DFA of $L$.
If $L$ is maximally atomic, then by Theorem \ref{thm:atomcomp} and Lemma \ref{lem:lw}, the subgroup of permutations in $T$ is $k$-set-transitive for $1 \le k \le n-1$, and hence is set-transitive; also, by Theorem \ref{thm:atomcomp}, $T$ contains a transformation of rank $n-1$.
This proves one direction of the theorem.

For the other direction, suppose the subgroup of permutations in $T$ is set-transitive and contains a transformation of rank $n-1$. 
By Theorem \ref{thm:atomnum}, $L$ has $2^n$ atoms.
By Theorem \ref{thm:atomcomp}, if $n \ge 4$ and $2 \le |S| \le n-2$ or $n = 3$ and $1 \le |S| \le 2$, then $A_S$ has maximal complexity.
By Lemma \ref{lem:gens}, $T$ contains all singular transformations  and hence all unitary transformations;  so  by Remark \ref{rem:unitary},  $A_S$ has maximal complexity if $|S| = 1$ or $|S| = n-1$.
The only remaining cases are $|S| = 0$ and $|S| = n$.

Let $\cA^\deter_S$ be the minimal DFA of $A_S$.
By Lemma \ref{lem:sintervals}, to show that $A_S$ has maximal complexity, it suffices to show that all atomic $S$-intervals of $L$ are reachable in $\cA^\deter_S$.
If $|S| = 0$, then $S = \emp$, and the atomic $\emp$-intervals of $L$ are those with type $(0,i)$ where $0 \le i \le n-1$.
The initial state of $\cA^\deter_\emp$ is $[[\emp,\emp]]$; thus a reachable state looks like $[[\delta_w(\emp),\Delta_w(\emp)]] = [[\emp,\coim\delta_w]]$ for some $w \in \Sig^*$.

Since $T$ contains all singular transformations, for all  $U\subset Q_n$, there exists $t \in T$ such that $\coim t = U$.
Hence for all $U \subset Q_n$, $[[\emp,U]]$ is reachable in $\cA^\deter_\emp$.
Thus all intervals of type $(0,i)$ are reachable, for $0 \le i \le n-1$.
By Lemma \ref{lem:sintervals}, $A_\emp$ has maximal complexity.
By a similar argument, when $|S| = n$, the atom $A_{Q_n}$ has maximal complexity.
Thus all $2^n$ atoms  have maximal complexity; this completes the proof. 
\qed

\section{Conclusions}
We have defined a new class of regular languages -- the maximally atomic languages -- and proven that a language of complexity $n$ is maximally atomic if and only if the transition semigroup of its minimal DFA is set-transitive and contains a transformation of rank $n-1$.
Since the set-transitive groups have been fully classified, it is easy to construct examples of maximally atomic languages and study them.
We have also derived a formula for the transition functions of \'atomata and minimal DFAs of atoms.
\smallskip 

\noindent{\bf Acknowledgements:} We thank a referee for giving many suggestions to improve our proofs.

\end{document}